\documentclass[Review,sageh,times]{sagej}

\usepackage{moreverb,url,subfigure}

\usepackage[colorlinks,bookmarksopen,bookmarksnumbered,citecolor=red,urlcolor=red]{hyperref}

\usepackage{graphicx,soul,framed,epsfig,mathptmx,times,
amsmath,amssymb,color,flushend,gensymb,subfigure,
fancyhdr,subfigure,lipsum,amsthm, amsfonts}

\usepackage{setspace}

\newtheorem{remark}{Remark}
\newtheorem{lemma}{Lemma}
\newtheorem{assumption}{Assumption}
\newtheorem{theorem}{Theorem }

\definecolor{shadecolor}{rgb}{1,0.8,0.3}

\newcommand\BibTeX{{\rmfamily B\kern-.05em \textsc{i\kern-.025em b}\kern-.08em
T\kern-.1667em\lower.7ex\hbox{E}\kern-.125emX}}

\def\volumeyear{2019}
\def\volumenumber{41}
\def\issuenumber{7}
\def\startpage{2039}
\def\endpage{2052}

\allowdisplaybreaks

\graphicspath{{Figures/}}

\def \journalname{\textbf{PREPRINT VERSION:} Transactions of the Institute of Measurement and Control}

\begin{document}

\runninghead{Kayacan}

\title{Sliding Mode Control for Systems with Mismatched Time-Varying Uncertainties via a Self-Learning Disturbance Observer}

\author{Erkan Kayacan}

\affiliation{Senseable City Laboratory, Computer Science $\&$ Artificial Intelligence Laboratory, Massachusetts Institute of Technology}

\corrauth{Erkan Kayacan, Senseable City Laboratory, Computer Science $\&$ Artificial Intelligence Laboratory, \\
Massachusetts Institute of Technology, Cambridge, MA 02139 USA.}
 
\email{erkank@mit.edu}

\begin{abstract}

This paper presents a novel Sliding Mode Control (SMC) algorithm to handle mismatched uncertainties in systems via a novel Self-Learning Disturbance Observer (SLDO). A computationally efficient SLDO is developed within a framework of feedback-error learning scheme in which a conventional estimation law and a Neuro-Fuzzy Structure (NFS) work in parallel. In this framework, the NFS estimates the mismatched disturbances and becomes the leading disturbance estimator while the former feeds the learning error to the NFS to learn system behavior. The simulation results demonstrate that the proposed SMC based on SLDO (SMC-SLDO) ensures the robust control performance in the presence of mismatched time-varying uncertainties when compared to SMC,  integral SMC (ISMC) and SMC based on a Basic Nonlinear Disturbance Observer (SMC-BNDO), and also remains the nominal control performance in the absence of mismatched uncertainties. Additionally, the SMC-SLDO not only counteracts mismatched time-varying uncertainties but also improve the transient response performance in the presence of mismatched time-invariant uncertainties. Moreover, the controller gain of the SMC-SLDO is required to be selected larger than the upper bound of the disturbance estimation error rather than the upper bound of the actual disturbance to guarantee the system stability which results in eliminating the chattering effects on the control signal.  

\end{abstract}

\keywords{Disturbance observers, mismatched uncertainty, neuro-fuzzy structures, learning algorithm, sliding mode control theory.}

\maketitle

\section{Introduction}

One of the most primary demands for control systems is to be insensitive to matched and mismatched disturbances \cite{Asl2017205, erkan2016, erkan2017, 31216683773, RAHMANI2016164, 6606388, KAYACAN20141, KAYACAN2012863}. Sliding Mode Control (SMC) has been studied and applied in industrial applications due to its advantage of robustness against parameter variations and disturbances \cite{Yu2009ie, PRECUP2017176}. Most of the studies on SMC focused only matched disturbances since traditional SMC is insensitive to disturbances, which appear in the same channel with the control inputs \cite{ZHANG201625, KAYACAN2016265}. However, disturbances may appear in different channels from control inputs, which result in mismatched uncertainties \cite{LI2016290}. There are many control methods for systems with mismatched uncertainties in literature, and they can be summed up into three main categories. 

In the first category, the mismatched uncertainties are assumed to be bounded with $H_{2}$ norm, and the stability of the systems with mismatched uncertainties was examined by using Riccati difference equation, adaptive, game theoretic, fuzzy control and LMI-based approaches \cite{Chang2009,PARK200772, Choi2007, CHANG2015276, CHANG2016258, ARQUB2014396, abo2014optimization}. However, this is not a realistic assumption in practice since mismatched uncertainties may have non-zero stead-state values. Therefore, these controllers have been combined with an adaptive model compensation having accurate online parameter estimation \cite{6983626, 6202698, 8409989, kayacanjfr}. These adaptive model compensations consider the neglected high-frequency dynamics and various nonlinearity effects in real-time.  In the second category, integral SMC (ISMC) is proposed in the presence of mismatched time-invariant disturbances \cite{Yu2009ie, Xu2016ie, RAHMANI2016, Rahmani20182}. However, it is acknowledged that an integral action can create undesirable effects, such as overshoots and  large settling times, and deteriorate the nominal control performance \cite{Cao2004ac,Xu2014ie}. In the third category, to overcome these important limitations, disturbance observer based control approaches have been proposed in which mismatched uncertainties are merged into one term, and the main goal is to derive a control law including the disturbance estimate to remove mismatched uncertainties effects on systems \cite{chen2003nonlinear, Ginoya2014,An2016ie, Li2015ie,RAHMANI2016117, KAYACAN2017276, Rahmani2018, Kayacanasj}. SMC based on a basic nonlinear disturbance observer (SMC-BNDO) has been developed to remain the nominal control performance in \cite{Yang2013}. The results have shown that the SMC-BNDO provides nominal performance recovery as well as better control performance as compared to the SMC and ISMC methods. However, if mismatched uncertainties are time-invariant, then BNDO causes bias estimates so that a DO is in want of the capability of estimating mismatched time-varying uncertainties. There exist very popular nonlinear observers, such as particle filtering, moving horizon estimator, and Kalman filter \cite{KAYACAN201578,7525615, KayacanRSS18}. Extended Kalman filter works well if noise on measurements is too small and the linear approximation for a system is available \cite{Haseltine2005}. Other estimators need large computational burden \cite{Daum2005}. Therefore, a computationally efficient disturbance observer is necessary to estimate time-varying uncertainties. 

Neuro-fuzzy structures (NFSs) were broady used for control and identification purposes. The basic idea behind the fusion of fuzzy logic theory and artificial neural networks is that the capability of learning from the input-output data for artificial neural networks and the capability of using expert knowledge for fuzzy logic. Gradient-based and evolutionary algorithms have widely been used for the tuning both the antecedent and consequent parts of neuro-fuzzy systems \cite{Ngo201674}. The main disadvantages of these approaches are that the partial derivatives are required and proving stability is questionable. Therefore, SMC theory-based learning algorithms have been proposed for NFSs \cite{Efe2000}. They ensure robustness and faster convergence rate than the conventional learning algorithms \cite{Kayacan2018chapter}. Furthermore, a control structure consisting of a conventional control law running in parallel with a neuro-fuzzy controller has been proposed for the control of systems and is called as feedback-error learning (FEL) algorithm  \cite{Gomi1993,Tolu2012,Sabahi20141} in which the conventional control law is utilized as the learning error to train the learning algorithms. Thanks to the learning process, the learning controllers take the overall control signal while the output of the conventional controller approaches zero. 

This paper describes SMC for systems with mismatched time-varying uncertainties via a self-learning disturbance observer to overcome the aforementioned limitations of existing SMC methods. The major contributions of this study are as follows: 
\begin{enumerate}
\item FEL algorithm is used to develop a self-learning disturbance observer (SLDO) for the first time.
\item In addition to the proof of the SMC theory-based learning algorithm, the overall system stability is proven by taking the SLDO dynamics into account. 
\item The last major contribution of this paper is to develop an SMC based on an SLDO. A novel control law and a novel sliding surface are derived by taking the disturbance and disturbance rate estimates into account to be able to drive the system states to the desired state trajectory in the presence of mismatched time-varying uncertainties. 
\end{enumerate} 

The minor contributions of this study are as follows: 
\begin{enumerate}
\item The SMC-SLDO remains the nominal control performance in the absence of mismatched disturbances and provides better control performance than SMC-BNDO in the presence of mismatched time-variant uncertainties. 
\item The controller gain of the SMC-SLDO is required to be larger than the upper bound of the disturbance estimation error rather than the upper bound of the disturbance in SMC and ISMC to guarantee the system stability; therefore, the chattering effects on the control signal are eliminated significantly. 
\end{enumerate}

The paper consists of six sections: The problem formulation is given in Section \ref{sec_prob_form}. The SLDO based on the combination of the BNDO and NFS is developed and its stability is proven in Section \ref{sec_SLDO}. The SMC-SLDO is designed, and the overall system stability by taking the SLDO dynamics into account is proven in Section \ref{sec_smcsldo}. The simulation results are presented in Section \ref{sec_simulation}. Finally, some conclusions are drawn in Section \ref{sec_conc}.

\section{Problem Formulation}\label{sec_prob_form}

A second-order nonlinear system with mismatched disturbances is formulated as follows:
\begin{eqnarray}\label{eq_nonlinearsystem}
\dot{\textbf{x} }& = & \textbf{g}_{1}(\textbf{x}) + \textbf{g}_{2}(\textbf{x}) u + \textbf{z} d(t)
\end{eqnarray}
where $\textbf{x}$ and $u$ are written instead of $\textbf{x}(t)$ and $u(t)$ for convenience, $u$ is the control input, $\textbf{x}=[x_{1},x_{2}]^{T}$ is the state vector, $\textbf{z}=[1,0]^{T}$ is the disturbance coefficients vector, $d(t)$ is the disturbance, $\textbf{g}_{1}(\textbf{x})=[x_{2}, a(\textbf{x})]^{T}$ and $\textbf{g}_{2}(\textbf{x})=[0, b(\textbf{x})]^{T}$ are nonlinear functions. The rotor speed of a permanent magnet synchronous motor and a single link manipulator with the elasticity of the joint are potential practical systems in the form of \eqref{eq_nonlinearsystem} \cite{Yang20132287, Ginoya2014}.

\subsection{Sliding Mode Control}\label{sec_smc}

The sliding surface and the control law are written as follows:
\begin{eqnarray}\label{eq_tsmc_slidingsurface}
s &=&  x_{2} + \lambda x_{1}  \\ \label{eq_tsmc_controllaw}
u &=& - b^{-1}(x) \Big( a(\textbf{x}) + \lambda x_{2} + k \textrm{sgn}(s) \Big)
\end{eqnarray}
where $\lambda$ is the slope of the sliding surface, $k$ is the controller gain, and they are positive, i.e., $\lambda, k>0$.

By combining \eqref{eq_nonlinearsystem}-\eqref{eq_tsmc_controllaw}, the time-derivative of the sliding surface is obtained as follows:
\begin{equation}\label{eq_tsmc_slidingsurface_dot}
\dot{s} =  - k \textrm{sgn}(s) + \lambda d(t)
\end{equation}
If the controller gain $k$ in \eqref{eq_tsmc_controllaw} is large enough, i.e., $k>\lambda d^{*}$ where $d^{*}$ is the upper bound of the disturbance, the sliding surface converges to zero in finite time. When the sliding surface converges to zero, the sliding motion is obtained as follows:
\begin{equation}\label{eq_tsmc_slidingsurfaceerror}
\dot{x}_{1} + \lambda x_{1} =  d(t)
\end{equation}

\begin{remark}\label{remark_tsmc}
If the system reaches the sliding surface in finite time, \eqref{eq_tsmc_slidingsurfaceerror} shows that the states cannot be driven to the desired state trajectory under the control law proposed in \eqref{eq_tsmc_controllaw}. This purports why the SMC is delicate to mismatched disturbance.
\end{remark}

\subsection{Integral Sliding Mode Control}\label{sec_ismc}

To handle mismatched uncertainty problem, an integral action is added to the sliding surface. The integral sliding surface and ISMC control law are written as follows \cite{slotine1991}:
\begin{eqnarray}\label{eq_ismc_slidingsurface}
s &=& x_{2}+ 2 \lambda  x_{1}  + \lambda^{2} \int^t_0 x_{1} \,dt \\ \label{eq_ismc_controllaw}
u &=& -b^{-1}(\textbf{x}) \Big( a(\textbf{x}) + 2 \lambda x_{2} + \lambda^{2} x_{1} + k \textrm{sgn}(s) \Big)
\end{eqnarray}

The time-derivative of the integral sliding surface is obtained by considering the nonlinear equation in  \eqref{eq_nonlinearsystem}, aforementioned integral sliding surface in \eqref{eq_ismc_slidingsurface} and ISMC control law in \eqref{eq_ismc_controllaw} as follows:
\begin{equation}\label{eq_ismc_slidingsurface_dot}
\dot{s} =  - k \textrm{sgn}(s) + 2 \lambda d(t) 
\end{equation}

If the controller gain $k$ in the control law in \eqref{eq_ismc_controllaw} is large enough, i.e., $k>2 \lambda d^{*}$ where $d^{*}$ is the upper bound of the disturbance, the integral sliding surface converges to zero in finite time. When the integral sliding surface converges to zero, the sliding motion is obtained as follows:
\begin{equation}\label{eq_ismc_slidingsurfaceerror}
\ddot{x}_{1} + 2 \lambda \dot{x}_{1} +  \lambda^{2} x_{1} =  \dot{d}(t)
\end{equation}
This purports that if the system reaches the integral sliding surface in finite time and the disturbance is time-invariant, i.e., $\dot{d}(t)=0$, the states can be driven to the desired state trajectory.

\begin{remark}\label{remark_ismc} Equation \eqref{eq_ismc_slidingsurfaceerror} shows that if disturbance is time-invariant, i.e., $\dot{d}(t)=0$, ISMC is robust against mismatched time-invariant disturbances so that the offset is removed. However, the integral action might deteriorate the performance, such as overshoot and worse control performance than the nominal control performance achieved by SMC in the absence of mismatched disturbances. 
\end{remark}

It is to be noted that there are several methods to overcome the limitation of ISMC in the absence of mismatched disturbances, such as an adequate initial condition in \cite{577594}. However, since it is a troublesome process in control design in practice, this issue is left out of account in this study.

\subsection{SMC based on a BNDO for Mismatched Time-Invariant Uncertainties}\label{sec_smcbndo}

It is required to estimate the disturbance $d(t)$ in \eqref{eq_nonlinearsystem} in real-time to achieve a robust control performance. The BNDO proposed in \cite{chen2003nonlinear,Chen2004,Yang2011} is written as below:
\begin{eqnarray}\label{eq_disturbanceobserver}
\dot{p} &=& - \textbf{l} \textbf{z} p - \textbf{l} \Big( \textbf{z} \textbf{l} \textbf{x} + \textbf{g}_{1}(\textbf{x}) + \textbf{g}_{2}(\textbf{x}) u \Big) \nonumber \\
\hat{d}_{BN} &=& p + \textbf{l} \textbf{x}
\end{eqnarray}
where $\hat{d}_{BN}$ is the estimated disturbance, $p$ is the internal state  and $\textbf{l} = [l_{1}, l_{2}]$ is the observer gain of the BNDO.

The sliding surface and control law based on the BNDO proposed in \cite{Yang2013,Yang20132287} are formulated as follows:
\begin{eqnarray}\label{eq_smcdo_slidingsurface}
s &=& x_{2} + \lambda x_{1} + \hat{d}_{BN} \\ \label{eq_smcdo_controllaw}
u &=& -b^{-1} (\textbf{x}) \Big( a(\textbf{x}) + \lambda (x_{2} + \hat{d}_{BN}) + k \textrm{sgn}(s) \Big)
\end{eqnarray}
Then, the time-derivative of the sliding surface is calculated taking \eqref{eq_nonlinearsystem}, \eqref{eq_smcdo_slidingsurface} and \eqref{eq_smcdo_controllaw} into account as
\begin{equation}\label{eq_smcdo_ssdot}
\dot{s} = -k \textrm{sgn}(s) + (\lambda + lz) e_{d}
\end{equation}
where $e_{d}= d(t) - \hat{d}_{BN}$ is the disturbance estimation error. If $k$ is large enough, i.e., $k> (\lambda + \textbf{l} \textbf{z}) e_{d}^{*}$ where $e_{d}^{*}$ is the upper bound of the disturbance estimation error, the sliding surface converges to zero in finite time. When the sliding surface converges to zero, the sliding motion is obtained as
\begin{equation}\label{eq_smcdo_slidingsurfaceerror}
 \dot{x}_{1} + \lambda x_{1} = e_{d}
\end{equation}
The disturbance rate estimate is derived by considering \eqref{eq_disturbanceobserver} as follows:
\begin{equation}
\dot{\hat{d}}_{BN} = \dot{p} + \textbf{l} \dot{\textbf{x}} 
\end{equation}
The equality of $\dot{p}$ in \eqref{eq_disturbanceobserver} is inserted into the equation above
\begin{eqnarray}
\dot{\hat{d}}_{BN} &=& - \textbf{l} \textbf{z} p - \textbf{l} \Big( \textbf{z} \textbf{l} \textbf{x} + \textbf{g}_{1}(\textbf{x}) + \textbf{g}_{2}(\textbf{x}) u \Big) + \textbf{l} \dot{\textbf{x}} \nonumber \\
  &=&  - \textbf{l} \textbf{z} \underbrace{(p+\textbf{l} \textbf{x})}_{\hat{d}_{BN}} + l \underbrace{ \Big( \dot{\textbf{x}} - \textbf{g}_{1}(\textbf{x}) - \textbf{g}_{2}(\textbf{x}) u \Big)}_{\textbf{z} d}
\end{eqnarray}
By taking \eqref{eq_nonlinearsystem} and \eqref{eq_disturbanceobserver} into account, the time-derivative of the estimated disturbance by the BNDO is calculated as follows:
\begin{equation}\label{eq_disturbanceobserver_dot_dBN}
\dot{\hat{d}}_{BN} = \textbf{l} \textbf{z}   e_{d}
\end{equation}
If the disturbance rate is inserted into aforementioned equation, the error dynamics for the BNDO are calculated as
\begin{equation}\label{eq_disturbanceobserver_error}
\dot{e}_{d} + \textbf{l} \textbf{z}  e_{d} = \dot{d} (t)
\end{equation}

\begin{assumption}\label{assumption_smcdo_1}
The disturbance rate is bounded and $\displaystyle \lim_{ t \to \infty} \dot{d}(t) = 0$.
\end{assumption}

If Assumption \ref{assumption_smcdo_1} is fulfilled, then \eqref{eq_smcdo_slidingsurfaceerror} and \eqref{eq_disturbanceobserver_error} are obtained as follows:
\begin{eqnarray}\label{eq_smcdo_system_1}
 \dot{x}_{1} + \lambda x_{1} & = & e_{d}  \\ \label{eq_smcdo_system_2}
\dot{e}_{d}  + \textbf{l} \textbf{z}   e_{d}  & = & 0
\end{eqnarray}

\begin{lemma}\label{lemma_smcdo_1}
\cite{chen2003nonlinear, KAYACAN2017276, Kayacanasj}: If $\textbf{l} \textbf{z}$ is positive, i.e., $\textbf{l} \textbf{z}>0$, the disturbance error dynamics in \eqref{eq_smcdo_system_2} approach zero exponentially. 
\end{lemma}

Lemma \ref{lemma_smcdo_1} purports that the disturbance estimation can follow the actual disturbance in \eqref{eq_nonlinearsystem} globally exponentially if Assumption \ref{assumption_smcdo_1} has been fulfilled.

\begin{lemma}\label{lemma_smcdo_2}
\cite{khalil1996nonlinear, KAYACAN2017276, Kayacanasj}: If a nonlinear system $\dot{\textbf{x}}=F(\textbf{x},u)$ is input-state stable and $\displaystyle \lim_{ t \to \infty} u(t) = 0$, then the state $\displaystyle \lim_{ t \to \infty} \textbf{x}(t) = 0$.
\end{lemma}

As indicated in Lemma \ref{lemma_smcdo_2}, this system in \eqref{eq_smcdo_system_1} is stable if the slope of the sliding surface is positive, i.e., $\lambda>0$. The state satisfies $\displaystyle \lim_{ t \to \infty} e_{d}(t) = 0$ and $\displaystyle \lim_{ t \to \infty} x_{1}(t) = 0$. 

\begin{remark} \label{remark_smcdo_1}
If disturbance is time-invariant as indicated in Assumption \ref{assumption_smcdo_1}, the SMC based on the BNDO (SMC-BNDO) is robust against time-invariant mismatched disturbances, and thus the offset is eliminated. Moreover, if there does not exist any disturbance, the SMC-BNDO remains the nominal control performance achieved by SMC.
\end{remark}

\begin{remark} \label{remark_smcdo_2}
If disturbance rate $\dot{d} (t)$ is not equal to zero, the BNDO  error dynamics cannot approach zero; therefore, the BNDO causes a bias.  The SMC-BNDO is not robust against mismatched time-varying disturbances. Similar observers were designed and the same drawback was reported in literature \cite{Chen2000,chen2003nonlinear,Yang2013}. 
\end{remark}

\section{Online Self-Learning Disturbance Observer}\label{sec_SLDO}

The BNDO fails to provide unbiased estimates in case of the existing of time-varying disturbances. For this reason, an SLDO, in which the BNDO works in series with FEL algorithm, is developed in this section. In FEL scheme, the conventional estimator works in parallel with an NFS as the schematic diagram of the SLDO is shown in Fig. \ref{fig_sldo_diagram}. Conventional estimation law establishes a sliding manifold and new estimation law is formulated as:
\begin{equation}\label{eq_sldo_estimationlaw}
\dot{\hat{d}}_{SL} = \dot{\hat{d}}_{BN} + \tau_{c} - \tau_{n}  
\end{equation}
where $\hat{d}_{BN}$, $\tau_{c}$ and $\tau_{n}$ denote the outputs of BNDO, conventional estimation law and NFS, respectively. The conventional estimation law is written below:
\begin{equation}\label{eq_sldo_tauc}
\tau_{c} = \ddot{\hat{d}}_{BN} + \textbf{l} \textbf{z} \dot{\hat{d}}_{BN} 
\end{equation}
where  $\textbf{l} \textbf{z}$ is positive, i.e., $\textbf{l} \textbf{z} >0$ as stated in Lemma \ref{lemma_smcdo_1}.
\begin{figure}[t!]
  \centering
  \includegraphics[width=4.5in]{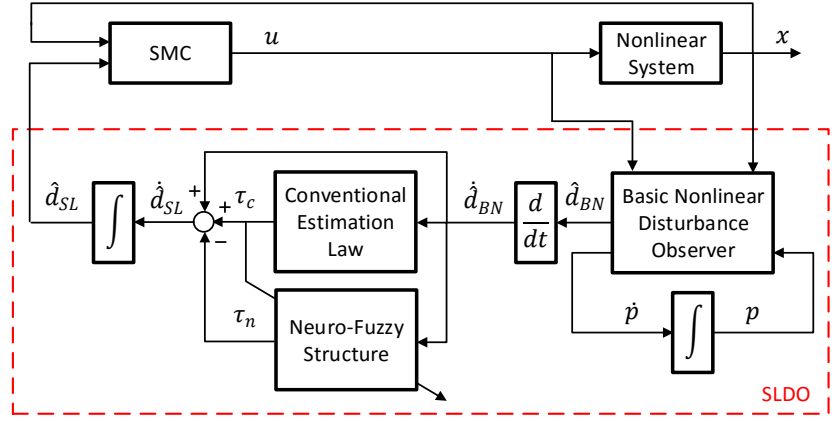}\\
  \caption{Schematic diagram for the control stucture and self-learning disturbance observer.}\label{fig_sldo_diagram}
\end{figure}

\subsection{Neuro-Fuzzy Structure}\label{sec_neuro-fuzzy}

An interval Takagi-Sugeno-Kang (TSK) fuzzy \emph{if-then} rule $R_{ij}$ is written as:\begin{equation}\label{eq_Rlinearfunction}
R_{ij}: \;\; \textrm{If} \; \xi_1 \; \textrm{is} \;\; A_{1i} \;\; \textrm{and} \; \xi_2 \; \textrm{is} \;\; A_{2j}, \;\; \textrm{then} \; f_{ij}=\Upsilon_{ij}
\end{equation}
where $\xi_1=\dot{\hat{d}}_{BN}$ and $\xi_2= \ddot{\hat{d}}_{BN}$ are respectively the inputs while $A_{1i}$ and $A_{2j}$ denote fuzzy sets corresponding to inputs. The zeroth-order function $f_{ij}$ is the output of the rules while $\Upsilon_{ij}$ is the consequent part of the rules and the total number of the rules are equal to $K=I \times J$ in which $I$ and $J$ are the total number of the inputs. Since the time-derivatives of the the output of the BNDO are the inputs of the neuro-fuzzy structure, the derivations with respect to time must be taken with a filter (e.g., similar to the transfer function $N/(1 + N/s)$) for noisy disturbances in real-time to obtain $\dot{\hat{d}}_{BN}$ and $\ddot{\hat{d}}_{BN}$.

The Gaussian membership functions are formulated as:
\begin{equation}\label{eq_mu1}
\mu_{1i}(\xi_1) = \exp\Bigg(-\bigg(\frac{\xi_{1}-c_{1i}}{\sigma_{1i}}\bigg)^2\Bigg)
\end{equation}
\begin{equation}\label{eq_mu2}
\mu_{2j}(\xi_2) = \exp\Bigg(-\bigg(\frac{\xi_{2}-c_{2j}}{\sigma_{2j}}\bigg)^2\Bigg)
\end{equation}
where $\mu_{1i}(\xi_1)$ is the membership values for the input 1 where $c_{1i}$ and $\sigma_{1i}$ are respectively the mean and standard deviation of this membership function for the input 1,  $\mu_{2j}(\xi_2)$ is the membership values for the input 2 where $c_{2j}$ and $\sigma_{2j}$ are respectively the mean and standard deviation of this membership function for the input 2. 

The output of each node is multiplied with all incoming signals  to calculate the firing strength of rules:
\begin{equation}\label{eq_wij}
w_{ij} = \mu_{1i}(\xi_1)  \mu_{2j}(\xi_2) 
\end{equation}
\eqref{eq_mu1} and \eqref{eq_mu2} are inserted into \eqref{eq_wij}, it is obtained as:
\begin{equation}\label{eq_wij_2}
w_{ij} = \exp\Bigg(-\bigg(\frac{\xi_{1}-c_{1i}}{\sigma_{1i}}\bigg)^2 -\bigg(\frac{\xi_{2}-c_{2j}}{\sigma_{2j}}\bigg)^2\Bigg)
\end{equation}

The output of the neuro-fuzzy structure is written as follows:
\begin{equation}\label{eq_taun}
\tau_n= \sum_{i=1}^{I}\sum_{j=1}^{J}f_{ij}\widetilde{w}_{ij} 
\end{equation}
where $\widetilde{{w}}_{ij}$ is the normalized firing strengths of the output signal of the neuron $ij$ which are calculated as follows:
\begin{eqnarray}\label{eq_wij_normalized}
\widetilde{w}_{ij} = \frac{w_{ij}}{\sum_{i=1}^{I}\sum_{j=1}^{J}w_{ij}}
\end{eqnarray}

The vectors are defined as:
\begin{eqnarray}
\widetilde{\textbf{W}} &=& [\widetilde{w}_{11} \; \widetilde{w}_{12} \dots \widetilde{w}_{21} \dots \widetilde{w}_{ij} \dots \widetilde{w}_{IJ}(t)]^{T} \nonumber \\
\textbf{F} &=& [f_{11} \; f_{12} \dots f_{21} \dots f_{ij} \dots f_{IJ}] \nonumber
\end{eqnarray}
where these normalized firing strengths take values between $0$ and $1$, i.e., $0<\widetilde{w}_{ij} \leq 1$. Additionally, $\sum_{i=1}^{I}\sum_{j=1}^{J}\widetilde{w}_{ij} = 1$.

\subsection{SMC Theory-based Learning Algorithm}\label{sec_SMClearning}

The sliding surface $\eta$ for the learning algorithm is formulated as follows:
\begin{equation}\label{eq_sldo_slidingfsurface}
\eta  \left( \tau_{c} \right)= \tau_{c} 
\end{equation}
where $\tau_{c}$ is utilized as the learning error to train the SMC theory-based learning algorithm, which adjusts the mean and standard deviation of the membership functions in the antecedent part, and also the consequent part in the NFS. The adaptation rules of the NFS are formulated as follows:
\begin{equation} \label{eq_c_1i}
\dot{c}_{1i} = \dot{\xi}_{1} + (\xi_{1} - c_{1i}) \alpha_{1} \textrm{sgn}\left(  \eta\right)
\end{equation}
\begin{equation} \label{eq_c_2j}
\dot{c}_{2j} = \dot{\xi}_{2} + (\xi_{2} - c_{2j}) \alpha_{1} \textrm{sgn}\left( \eta\right)
\end{equation}
\begin{equation}\label{eq_sigma_1i}
\dot{\sigma}_{1i} = - \bigg( \sigma_{1i} + \frac{ (\sigma_{1i} )^3}{(\xi_{1} - c_{1i})^2} \bigg) \alpha_{1}  \textrm{sgn}\left( \eta \right)
\end{equation}
\begin{equation}\label{eq_sigma_2j}
\dot{\sigma}_{2j} = - \bigg( \sigma_{2j} + \frac{ (\sigma_{2j} )^3}{(\xi_{2} - c_{2j})^2} \bigg) \alpha_{1}  \textrm{sgn}\left( \eta \right)
\end{equation}
\begin{equation}\label{eq_f_ij}
\dot{f}_{ij} =-\frac{\widetilde{W} }{\widetilde{W}^T\widetilde{W}}\alpha_{2} sgn( \eta )
\end{equation}
where $\alpha_{1}$, $\alpha_{2}$ are the learning rates and positive, i.e., $\alpha_{1}, \alpha_{2}>0$.

\begin{assumption}\label{ass_boundedBNSL}
The estimated disturbances by the BNDO and SLDO, and their first and second time derivatives are assumed to be bounded:
\begin{eqnarray}
\mid \hat{d}_{BN} \mid < \hat{d}^{*}_{BN},  \quad \mid \dot{\hat{d}}_{BN}  \mid < \dot{\hat{d}}^{*}_{BN},  \quad \mid \ddot{\hat{d}}_{BN}  \mid < \ddot{\hat{d}}^{*}_{BN}, \quad
\mid \hat{d}_{SL} \mid < \hat{d}^{*}_{SL},  \quad \mid \dot{\hat{d}}_{SL}  \mid < \dot{\hat{d}}^{*}_{SL},  \quad \mid \ddot{\hat{d}}_{SL}  \mid < \ddot{\hat{d}}^{*}_{SL} 
\end{eqnarray}
where $\hat{d}^{*}_{BN}$, $\dot{\hat{d}}^{*}_{BN}$, $\ddot{\hat{d}}^{*}_{BN}$, $\hat{d}^{*}_{SL}$, $\dot{\hat{d}}^{*}_{SL}$ and $\ddot{\hat{d}}^{*}_{SL}$ are considered as positive constants.
\end{assumption}

\begin{theorem}[]\label{theorem1}
If adaptations rules are formulated as in \eqref{eq_c_1i}-\eqref{eq_f_ij} and the learning rate $\alpha_{2}$ is large enough, i.e., $\alpha_{2} > \ddot{\hat{d}}^{*}_{BN}  + \ddot{\hat{d}}^{*}_{SL}$, this guarantees that $\tau_{c}$ approaches zero.
\end{theorem}

\begin{proof}
The Lyapunov function candidate is formulated below:
\begin{equation}\label{eq_sldo_Lyapunov}
V = \frac{1}{2} \tau_{c}^{2} 
\end{equation}
The time-derivative of the Lyapunov function in\eqref{eq_sldo_Lyapunov} is taken as follows: 
\begin{equation}\label{eq_sldo_Lyapunov_dot} 
\dot{V} =   \tau_{c} \dot{\tau}_{c} 
\end{equation}
If the time-derivative of \eqref{eq_sldo_estimationlaw} is inserted into \eqref{eq_sldo_Lyapunov_dot}
\begin{equation}\label{eq_sldo_Lyapunov_dot2} 
\dot{V} =   \tau_{c} (\dot{\tau}_{n} - \ddot{\hat{d}}_{BN}  + \ddot{\hat{d}}_{SL} ) 
\end{equation}
Then, the calculation of $\dot{\tau}_{n}$ in \eqref{eq_taun_dot_3} (see Appendix) is inserted into \eqref{eq_sldo_Lyapunov_dot2} by taking \eqref{eq_sldo_slidingfsurface} into account, 
\begin{equation}\label{eq_sldo_Lyapunov_dot4}
\dot{V} =   \tau_{c} \Big( -\alpha_{2}  \textrm{sgn}\left( \tau_{c}\right)  -\ddot{\hat{d}}_{BN}  + \ddot{\hat{d}}_{SL} \Big) 
\end{equation}
As stated in Assumption \ref{ass_boundedBNSL}, $\ddot{\hat{d}}_{SL}$ and $\ddot{\hat{d}}_{BN} $  are respectively bounded by $\ddot{\hat{d}}^{*}_{SL}$ and $\ddot{\hat{d}}^{*}_{BN}$; therefore,  \eqref{eq_sldo_Lyapunov_dot4} is re-written:
\begin{equation}\label{eq_sldo_Lyapunov_dot5}
\dot{V} <  \mid \tau_{c} \mid ( - \alpha_{2} + \ddot{\hat{d}}^{*}_{BN}   + \ddot{\hat{d}}^{*}_{SL} ) 
\end{equation}
As stated in Theorem \ref{theorem1}, if the learning rate $\alpha_{2}$ is large enough, i.e., $\alpha_{2}>  \ddot{\hat{d}}^{*}_{BN}  + \ddot{\hat{d}}^{*} _{SL} $, then the time-derivative of the Lyapunov function is negative; therefore, the SMC theory-based learning algorithm is globally asymptotically stable and $\tau_{c}$ converges to zero.
\end{proof}

\subsection{Stability Analysis of the SLDO}\label{sec_stability_SLDO}

\begin{assumption}\label{ass_boundedddotd}
The acceleration of the actual disturbance is assumed to be bounded:
\begin{equation}
\mid \ddot{d} \mid < \ddot{d}^{*}
\end{equation}
where $\ddot{d}^{*}$ is considered as a positive constant. 
\end{assumption}

\begin{theorem}[]\label{theorem2}
The estimation law in \eqref{eq_sldo_estimationlaw} is utilized as a DO, the SLDO error dynamics are asymptotically stable if $lz$ is positive as stated in Lemma \ref{lemma_smcdo_1}, i.e., $\textbf{l}\textbf{z}>0$ and the learning rate of the NFS $\alpha_{2}$ is large enough, $\alpha_{2} > \ddot{d}^{*} $.
\end{theorem}

\begin{proof}

The Lyapunov function candidate is proposed as below:
\begin{equation}\label{eq_observer_Lyapunov}
V= \frac{1}{2}  \eta^{2} 
\end{equation}
The time-derivative of the Lyapunov function candidate is taken as below:
\begin{equation}\label{eq_observer_Lyapunov_d} 
\dot{V} =  \eta \dot{\eta} 
\end{equation}

If  $\tau_{c}$ in \eqref{eq_sldo_tauc} is inserted into \eqref{eq_sldo_slidingfsurface}, the sliding surface is obtained:
\begin{eqnarray}\label{eq_slidingfsurface_2}
\eta \left( \dot{\hat{d}}_{BN}, \ddot{\hat{d}}_{BN} \right)  =\ddot{\hat{d}}_{BN} + \textbf{l} \textbf{z} \dot{\hat{d}}_{BN}  
\end{eqnarray}
where $\textbf{l} \textbf{z}$ is the slope of the sliding surface. 
The time-derivative of the sliding surface is taken as follows:
\begin{equation}\label{eq_slidingfsurface_dot}
\dot{\eta} = \dddot{\hat{d}}_{BN} + \textbf{l} \textbf{z} \ddot{\hat{d}}_{BN}  
\end{equation}

If \eqref{eq_slidingfsurface_dot} is inserted into \eqref{eq_observer_Lyapunov_d}, the time-derivative of the Lyapunov function candidate is re-written as below:  
\begin{eqnarray}\label{eq_observer_Lyapunov_dd} 
 \dot{V} &=&  \eta ( \dddot{\hat{d}}_{BN} + \textbf{l} \textbf{z} \ddot{\hat{d}}_{BN}  )
\end{eqnarray}
The time-derivative of the Lyapunov function candidate is re-written by taking \eqref{eq_disturbanceobserver_dot_dBN} into account
\begin{eqnarray}\label{eq_observer_Lyapunov_d1} 
 \dot{V} &=&  \eta \Big( \textbf{l} \textbf{z} \ddot{e}_{d} + (\textbf{l} \textbf{z})^{2} \dot{e}_{d}  \Big)  
\end{eqnarray}

The SLDO law in \eqref{eq_sldo_estimationlaw} is re-written considering \eqref{eq_sldo_tauc}:
\begin{eqnarray}\label{eq_obs_estimation_law_self_stability}
\dot{\hat{d}}_{SL} = (1 + \textbf{l} \textbf{z}) \dot{\hat{d}}_{BN}  + \ddot{\hat{d}}_{BN}  - \tau_{n}
\end{eqnarray}
The SLDO error dynamics are calculated by adding the actual disturbance rate $\dot{d}$ into in \eqref{eq_obs_estimation_law_self_stability} and by taking \eqref{eq_disturbanceobserver_dot_dBN} into account:
\begin{eqnarray}\label{eq_obs_error_self_1}
 \dot{d} - \dot{\hat{d}}_{SL}  &=&  -(1 +  \textbf{l} \textbf{z}) \dot{\hat{d}}_{BN}  - \ddot{\hat{d}}_{BN}  + \tau_{n}+ \dot{d} \nonumber \\
\dot{e}_{d}  &=&  \frac{  - \textbf{l} \textbf{z} (1 + \textbf{l} \textbf{z}) e_{d}  + \tau_{n} + \dot{d} } { 1+ \textbf{l}  \textbf{z} }
\end{eqnarray}
The time-derivative of \eqref{eq_obs_error_self_1} is taken and obtained as below:
\begin{equation}\label{eq_obs_error_self_2}
\ddot{e}_{d}  =   \frac{  - \textbf{l} \textbf{z} (1 + \textbf{l} \textbf{z}) \dot{e}_{d}  + \dot{\tau}_{n} + \ddot{d} } { 1+ \textbf{l} \textbf{z} }
\end{equation}
$\dot{\tau}_{n}=- \alpha_{2} sgn(\eta)$ in \eqref{eq_taun_dot_3} is inserted into \eqref{eq_obs_error_self_2};
\begin{equation}\label{eq_obs_error_self_3}
\ddot{e}_{d}  =   \frac{ -\textbf{l} \textbf{z}(1 + \textbf{l} \textbf{z}) \dot{e}_{d}  - \alpha_{2} sgn(\eta) + \ddot{d}} {1 + \textbf{l} \textbf{z}}
\end{equation}

If the second time-derivative of the disturbance estimation error \eqref{eq_obs_error_self_3} is inserted into \eqref{eq_observer_Lyapunov_d1}:
\begin{eqnarray}\label{eq_observer_Lyapunov_d2}
\dot{V} & = & \eta \textbf{l} \textbf{z} \Big( \frac{ -\textbf{l} \textbf{z}(1 + \textbf{l}\textbf{z}) \dot{e}_{d}  - \alpha_{2} sgn(\eta) + \ddot{d}} {1 + \textbf{l} \textbf{z}} +  \textbf{l}  \textbf{z} \dot{e}_{d} \Big)
\end{eqnarray}
As stated in Assumption \ref{ass_boundedddotd}, $\ddot{d}$ is upper bounded by $\ddot{d}^{*}$:
\begin{eqnarray}\label{eq_observer_Lyapunov_d3}
 \dot{V}  &<&  \mid \eta \mid  \textbf{l} \textbf{z} \frac{(-  \alpha_{2} + \ddot{d}^{*})}{1 + \textbf{l} \textbf{z}} +\eta \dot{e}_{d}  \textbf{l} \textbf{z}  \underbrace{  \Big( - \textbf{l} \textbf{z} \frac{ (  1 + \textbf{l} \textbf{z}) }{1 + \textbf{l} \textbf{z}} + \textbf{l} \textbf{z} \Big) }_{0} \nonumber \\
 \dot{V}  &<&  \mid \eta \mid  \textbf{l} \textbf{z} \frac{(-  \alpha_{2} + \ddot{d}^{*})}{1 + \textbf{l} \textbf{z}}  
\end{eqnarray}

As stated in Theorem \ref{theorem2}, if $\textbf{l} \textbf{z}>0$ and $\alpha_{2} >\ddot{d}^{*}$, then the time-derivative of the Lyapunov function is negative, i.e., $\dot{V}<0$, so that the sliding surface $\eta$ converges to zero in finite time and the SLDO is globally asymptotically  stable.

\end{proof}

\begin{remark}\label{remark_sldo_timevarying}
The main advantage of the SLDO is to ensure the stability in the presence of time-varying disturbances.
\end{remark}

\begin{remark}\label{remark_sldo_timevaryingd}
The time-derivative of the disturbance is assumed to bounded with a finite changing rate in Assumption \ref{ass_boundedddotd}. Therefore, if the time-derivative of the disturbance is not oscillating significantly in practice, the SLDO algorithm is a convenient observer since we do not need to select large learning rate $\alpha_{2}$. 
\end{remark}

\section{Novel SMC-Based on an SLDO for Mismatched Time-Varying Uncertainties}\label{sec_smcsldo}

The control structure of the novel SMC-SLDO is illustrated in Fig. \ref{fig_sldo_diagram}. In this structure, the mismatched uncertainty is estimated by the SLDO, and the estimated disturbance is fed to the SMC controller. The SMC controller based on the system model and disturbance estimate generates a control signal for the system.

The sliding surface as proposed in \eqref{eq_smcdo_slidingsurface} is written as follows:
\begin{equation}\label{eq_smcsldo_slidingsurface}
s = x_{2} + \lambda x_{1} + \hat{d}_{SL} 
\end{equation}

The equivalent control law is derived by equaling the time-derivative of the sliding surface above to zero as follows:
\begin{equation}\label{eq_smcsldo_equivalentcontrollaw}
u_{eq} = -b^{-1} (\textbf{x}) \Big( a(\textbf{x}) + \lambda (x_{2} + \hat{d}_{SL}) + \dot{\hat{d}}_{SL} \Big)
\end{equation}

A novel control law by taking the estimated disturbance, estimated disturbance rate and discontinuous term into account is obtained as follows:
\begin{equation}\label{eq_smcsldo_controllaw}
u = -b^{-1} (\textbf{x}) \Big( a(\textbf{x}) + \lambda (x_{2} + \hat{d}_{SL}) + \dot{\hat{d}}_{SL}+ k \textrm{sgn}(s) \Big)
\end{equation}

Then, the time-derivative of the sliding surface is calculated by taking \eqref{eq_nonlinearsystem}, \eqref{eq_smcsldo_slidingsurface} and \eqref{eq_smcsldo_controllaw} into account as follows:
\begin{equation}\label{eq_smcsldo_slidingsurfacedot}
\dot{s} = -k \textrm{sgn}(s) + \lambda e_{d} 
\end{equation}
where $e_{d} = d(t) - \hat{d}_{SL}$ is the disturbance estimation error of the SLDO.

\begin{assumption}\label{ass_bounded_ed}
The disturbance error is assumed to be bounded:
\begin{equation}
\mid e_{d} \mid < e_{d}^{*}
\end{equation}
where $ e_{d}^{*}$ is considered as a positive constant. 
\end{assumption}

\begin{theorem}[Stability of overall system]\label{theorem3}
If the SLDO is employed to estimate the disturbance in \eqref{eq_nonlinearsystem}, the control law in \eqref{eq_smcsldo_controllaw} is applied to the nonlinear system in \eqref{eq_nonlinearsystem}, and the controller coefficient $k$ is large enough, i.e., $k > \lambda e_{d}^{*}$, then the overall system is  asymptotically stable. 
\end{theorem}

\begin{proof}
The Lyapunov function candidate is formulated as below:
\begin{eqnarray}\label{eq_smcsldo_Lyapunov}
V = \frac{1}{2} s^{2} 
\end{eqnarray}
The time-derivative of the Lyapunov function in \eqref{eq_smcsldo_Lyapunov} is taken as below:
\begin{eqnarray}\label{eq_smcsldo_Lyapunov_dot} 
\dot{V} =   s \dot{s}
\end{eqnarray}
Then, \eqref{eq_smcsldo_slidingsurfacedot} is inserted into \eqref{eq_smcsldo_Lyapunov_dot}
\begin{eqnarray}\label{eq_smcsldo_Lyapunov_dot2} 
\dot{V} =   s ( -k \textrm{sgn}(s) + \lambda e_{d})
\end{eqnarray}
As stated in Assumption \ref{ass_bounded_ed}, $e_{d}$ is bounded by $e_{d}^{*}$, \eqref{eq_smcsldo_Lyapunov_dot2} is calculated as below:
\begin{eqnarray}\label{eq_smcsldo_Lyapunov_dot3}
\dot{V} <   \mid s \mid  ( -k + \lambda e_{d}^{*} ) 
\end{eqnarray}
As remarked in Theorem \ref{theorem2}, the controller coefficient $k$ is large enough, i.e., $k > \lambda e_{d}^{*}$, then the time-derivative of the Lyapunov function is negative, i.e., $\dot{V}<0$; therefore, overall system is globally asymptotically stable and the sliding surface $s$ will converge to zero in finite time.

If the sliding surface converges to zero, i.e., $s=0$, this condition results in
\begin{equation}\label{eq_smcsldo_slidingsurfaceerror}
 \dot{x}_{1} + \lambda x_{1} = e_{d}
\end{equation}

The SLDO error dynamics go to zero as remarked in Theorem \ref{theorem2}. Therefore, Lemma \ref{lemma_smcdo_2} purports that the states can be driven to the desired state trajectory, i.e., $\displaystyle \lim_{ t \to \infty} e_{d}(t) = 0$, and $\displaystyle \lim_{ t \to \infty} x_{1}(t) = 0$.
\end{proof}

\begin{remark}\label{remark_smcsldo}
If there does not exist any disturbance, i.e.,$d(t)=\dot{d}(t)=0$, then the disturbance and disturbance rate  estimates in the sliding surface \eqref{eq_smcsldo_slidingsurface} and control law \eqref{eq_smcsldo_controllaw} are equal to zero, i.e., $\hat{d}_{SL}=\dot{\hat{d}}_{SL}=0$. This results in SMC; therefore, it can remain the nominal control performance.
\end{remark}

The implementation of the proposed method is summarized step by step here:
\begin{enumerate}
\item Design a sliding surface considering where the mismatched disturbance appears. 
\item Derive a control law that makes the time-derivative of the sliding surface zero, i.e., $\dot{s}=0$. 
\item Design a BNDO considering where the mismatched disturbance appears. 
\item Design a conventional estimation law considering the observer gain of the BNDO and the coefficient of the disturbance vector in the system model. 
\end{enumerate}

\section{Simulation Studies}\label{sec_simulation}

Motivated by \cite{Yang2013,Li2012generalized}, the following system is considered for simulation studies to benchmark control strategies:
\begin{eqnarray}\label{eq_cart}
\dot{x}_{1} & = & x_{2} + d(t) \nonumber \\
\dot{x}_{2} & = & - x_{1} - x_{2} + (x_{2})^{2} \cos{(x_{1})} + e^{x_1} + u
\end{eqnarray}
where $x_{1}$ and $x_{2}$ represent the states, $u$ is the input, $d$ is the disturbance, $a(\textbf{x})=- x_{1} - x_{2} + (x_{2})^{2} \cos{(x_{1})} + e^{x_1}$, $b(\textbf{x})=1$ and $\textbf{z}=[1, 0]^{T}$. The sampling time and the number of membership functions are set to $0.001$ second and $I = J = 3$, respectively.

\subsection{Scenario 1: General Comparison}\label{sec_scenario1}

The initial conditions on the states are set to $\textbf{x}(0)=[0.5, -0.5]^{T}$. To evaluate control performances of controllers in the absence and presence of the mismatched disturbances, disturbance is not imposed on the system at $t=0-10$ second at the beginning of the simulation. Then, a step external disturbance $d(t) = 0.3$ is imposed at $t =10$ second as a time-invariant disturbance. At the last step, a multisine disturbance $d(t)=0.15 \times \big(\sin{(t)} + \sin{(2t)}  \big)$ is imposed at $t=20$ second as a time-varying disturbance.

As remarked in Sections \ref{sec_smc} and \ref{sec_ismc}, the controller gain $k$ must be respectively larger than $\lambda d$ and $2 \lambda d$ for SMC and ISMC methods. Therefore, the controller gain must be selected considering the maximum disturbance value and the slope of the sliding surface.  Since the maximum disturbance value is equal to $0.6$, the controller gain $k$ must be larger than $0.6 \lambda$ i.e., $k > 0.6 \lambda $. Consequently, the controller gain and the slope of the sliding surface of SMCs are respectively set to $k=6.5$ and $\lambda=5$. 

As stated in Theorem  \ref{theorem2}, the learning rate of the SLDO $\alpha_{2}$ must be larger than the upper bound of the second derivative of the actual disturbance, i.e., $\alpha_{2} > \ddot{d}^{*}$, to guarantee the stability of the SLDO. Therefore, the learning rate $\alpha_{2}$ must be larger than the maximum value of the second derivative of the actual disturbance, i.e., $\alpha_{2} > 0.6$. Moreover, the neuro-fuzzy algorithm is more sensitive to the changes in the antecedent part than the changes in the consequent part so that $\alpha_{1}$ must be selected smaller than $\alpha_{2}$. Consequently, the observer gain vector of the BNDO is set to $l=[5, 0]^{T}$, and the learning rates $\alpha_{1}$, $\alpha_{2}$ of the SLDO are respectively set to $0.01$ and $1$. 

The states responses $x_{1}, x_{2}$ are shown in Figs. \ref{fig_x1}-\ref{fig_x2}, respectively. The SMC gives robust control performance in the absence of mismatched disturbance while it is delicate to mismatched disturbance as remarked in Remark \ref{remark_tsmc}. The ISMC  and SMC-BNDO provide robust control performance in the absence of mismatched disturbance and in the presence of mismatched time-invariant disturbance while they are sensitive to mismatched time-varying disturbance. However, ISMC causes large overshoot, and settling and rise times as stated in Remark \ref{remark_ismc} while SMC-BNDO remains the nominal control performance of SMC as noted in Remark \ref{remark_smcdo_1} in the absence of mismatched disturbance. The reason why SMC-BNDO is not robust against mismatched time-varying disturbance is that the BNDO fails to estimate time-varying disturbances as remarked in Remark \ref{remark_smcdo_2}. The SMC-SLDO gives robust control performance in the absence and presence of mismatched disturbances.  Moreover, it remains the nominal control performance of the SMC in the absence of disturbance as remarked in Remark \ref{remark_smcsldo}.  

The control signals for SMC, ISMC, SMC-BNDO and SMC-SLDO are shown in Fig. \ref{fig_control_signals} and they have severe high-frequency oscillations, i.e., chattering effects. The actual and estimated disturbances are shown in Fig. \ref{fig_d}. The BNDO fails to estimate  time-varying disturbance while the SLDO succeed to estimate time-varying disturbance as stated in Remark \ref{remark_sldo_timevarying}. Through the learning process by the NFS, the NFS can take the responsibility of the total estimation signal while the conventional estimation signal $\tau_{c}$ approaches to zero as shown in Fig. \ref{fig_estimation_signals}. The phase portraits for the SMC-SLDO are shown in Fig. \ref{fig_x1x1dot}. The reaching of the sliding surface and the sliding mode subsequently can be seen.

\begin{figure}[h!]
\centering
\subfigure[ ]{
\includegraphics[width=0.99\textwidth]{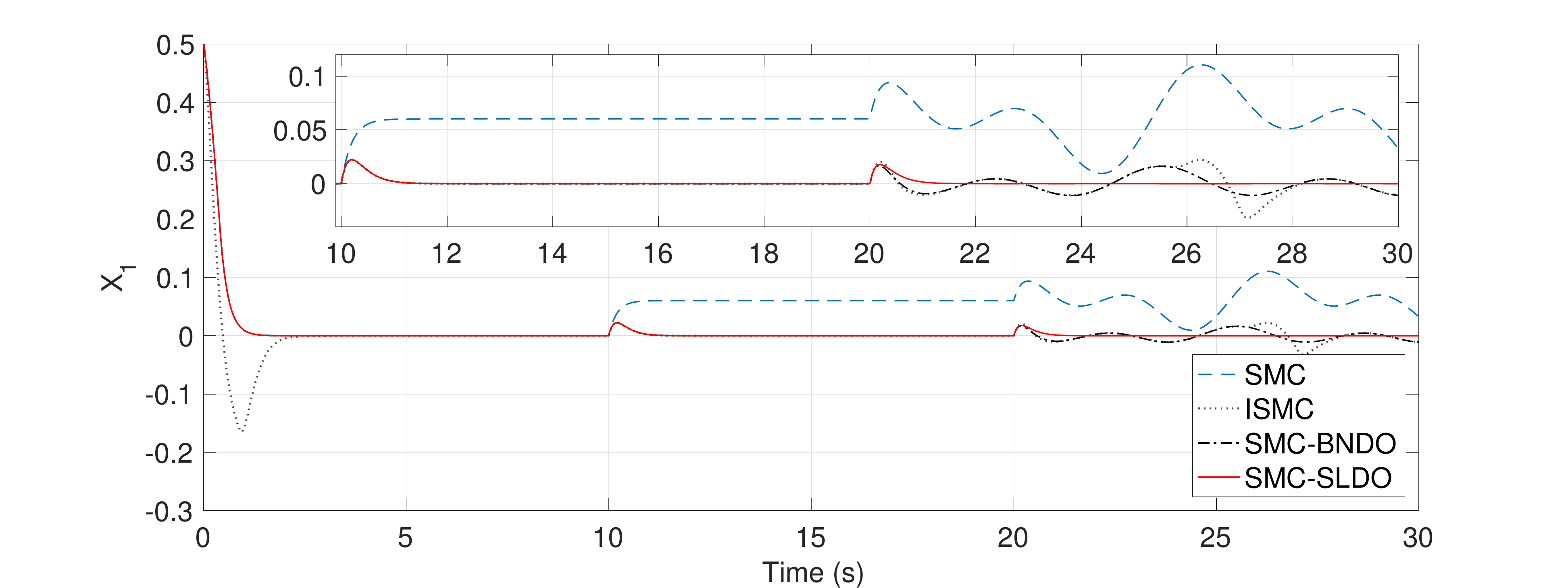}
\label{fig_x1}
}
\subfigure[ ]{
\includegraphics[width=0.99\textwidth]{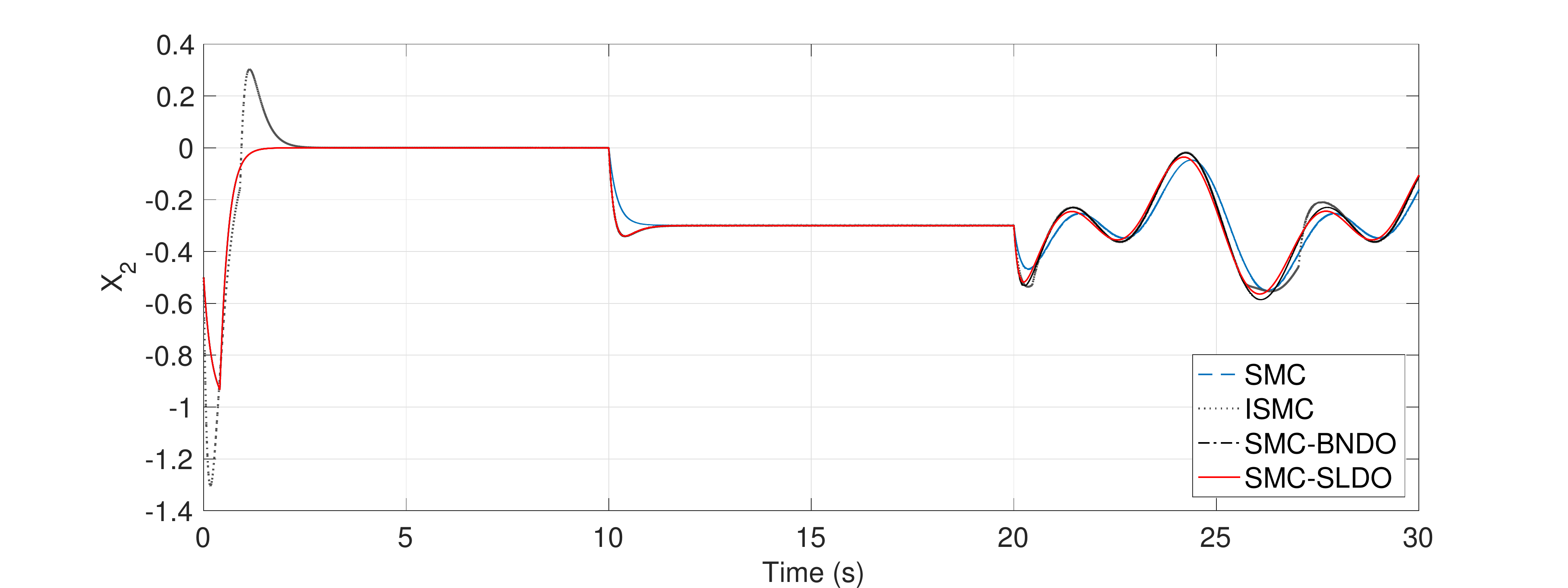}
\label{fig_x2}
}
\subfigure[ ]{
\includegraphics[width=0.99\textwidth]{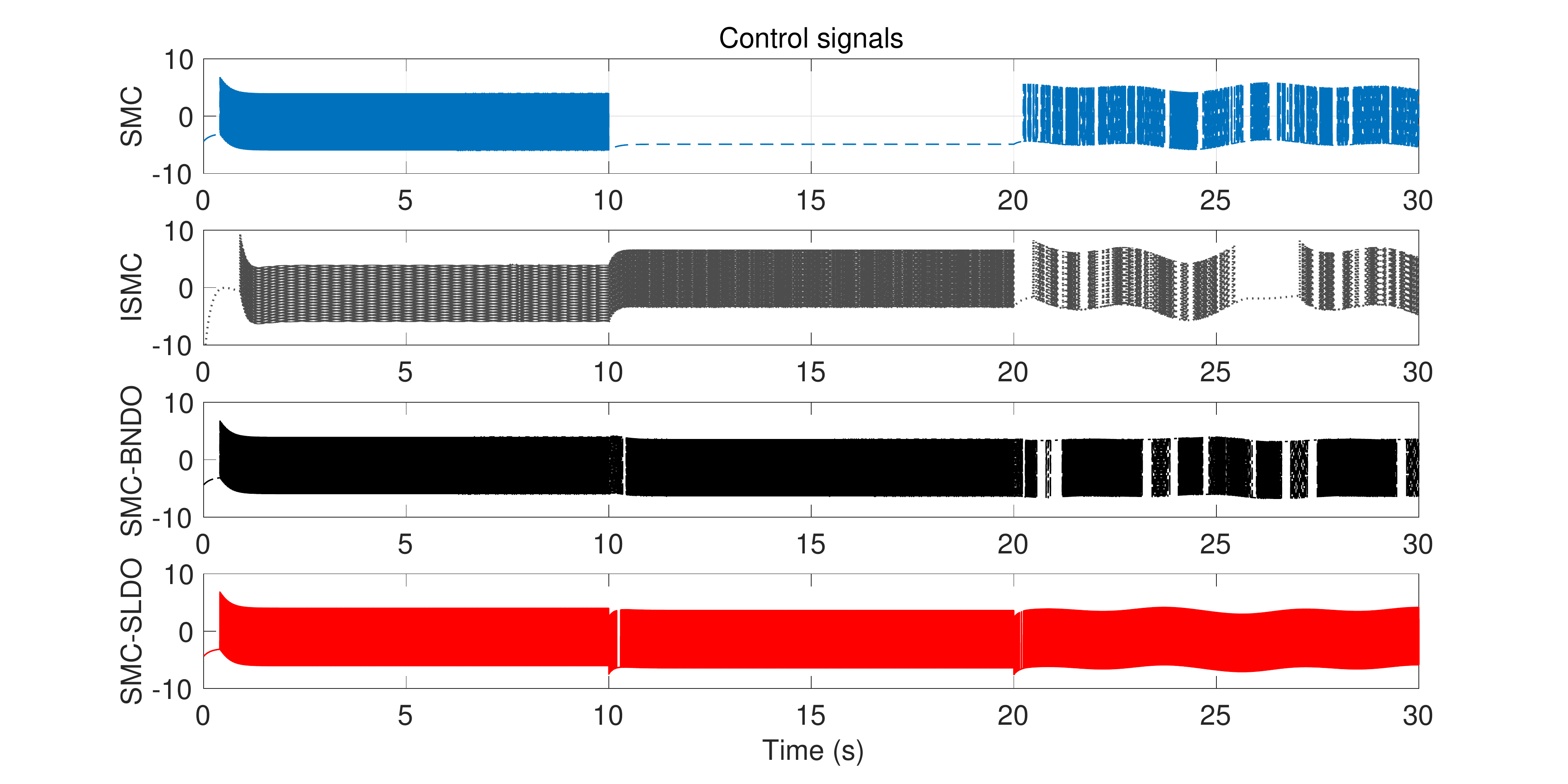}
\label{fig_control_signals}
}
\caption[Optional caption for list of figures]{ Scenario 1: General comparison (a) Responses of state $x_{1}$ (b) Responses of state $x_{2}$ (c) Control signals  }
\label{fig_1}
\end{figure}

\begin{figure}[h!]
\centering
\subfigure[ ]{
\includegraphics[width=0.99\textwidth]{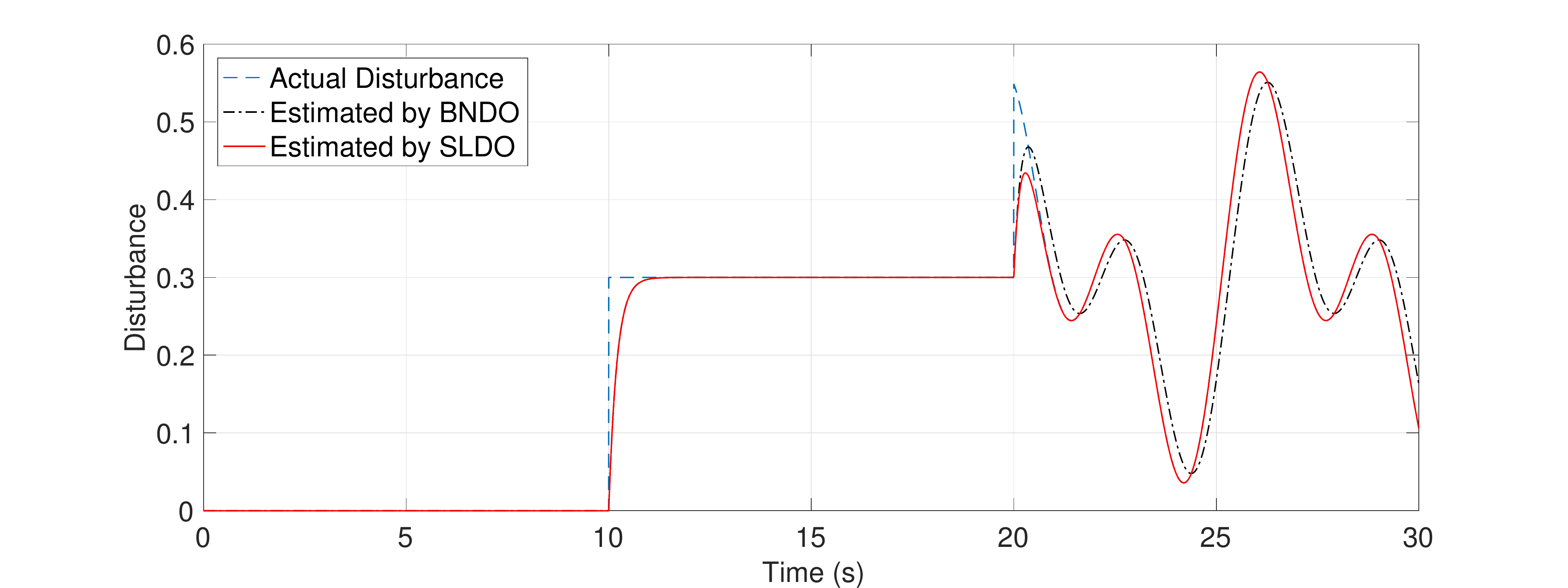}
\label{fig_d}
}
\subfigure[ ]{
\includegraphics[width=0.99\textwidth]{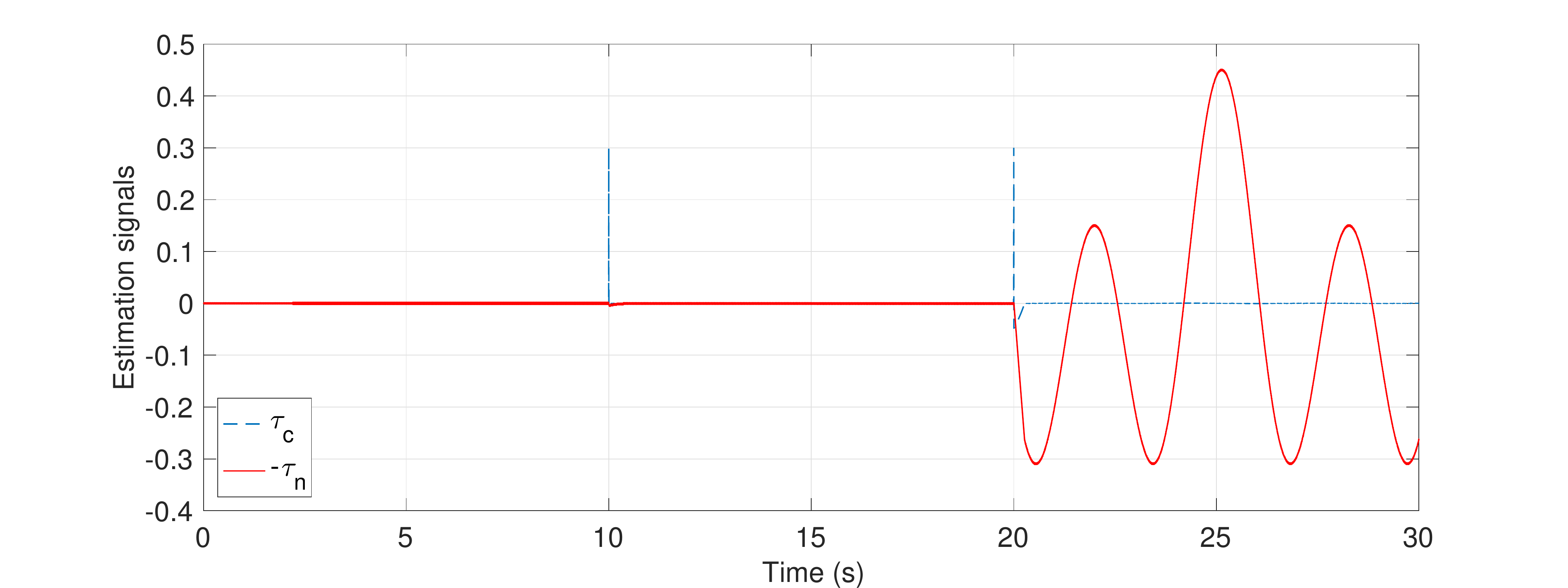}
\label{fig_estimation_signals}
}
\subfigure[ ]{
\includegraphics[width=0.99\textwidth]{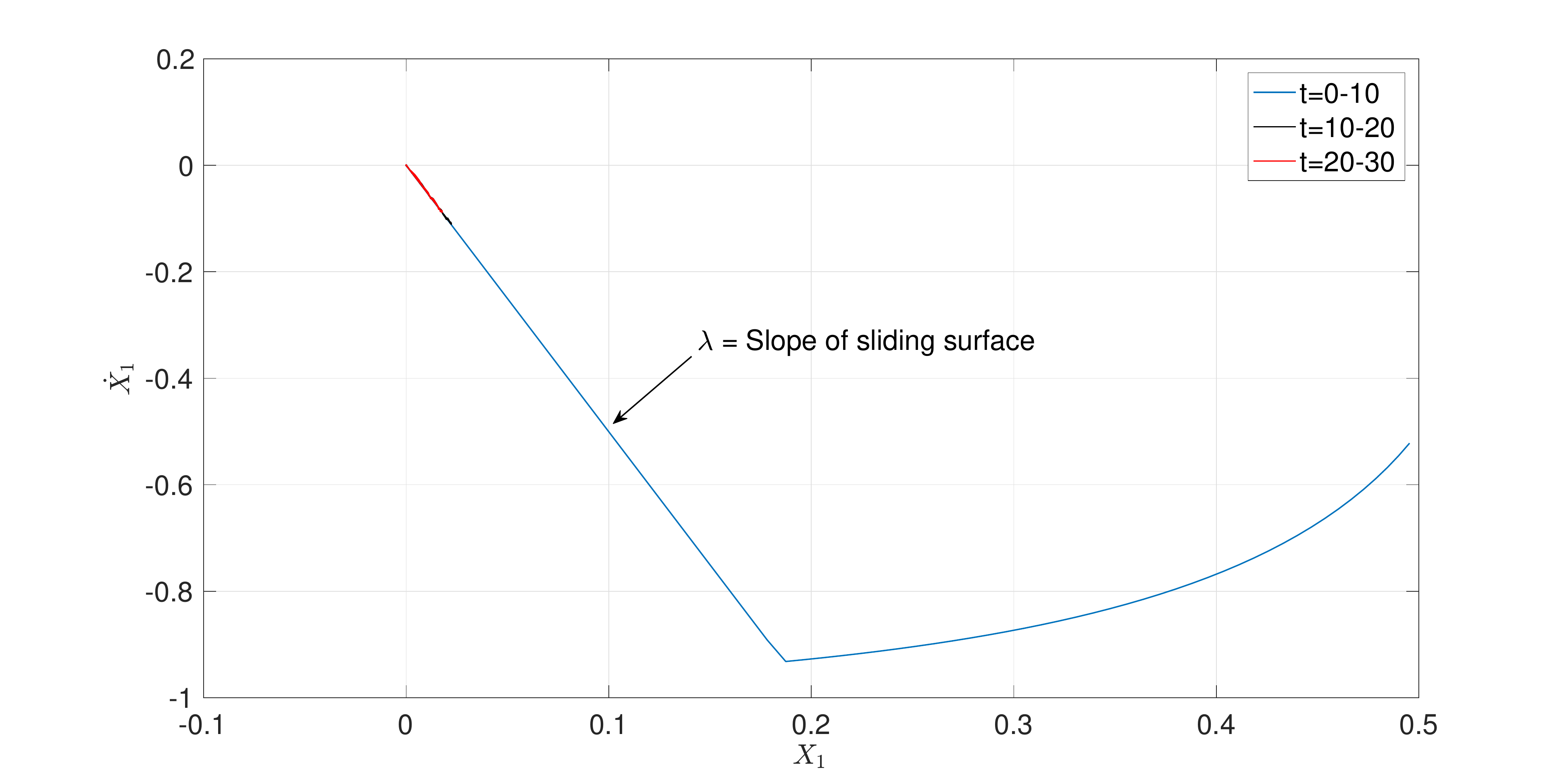}
\label{fig_x1x1dot}
}
\caption[Optional caption for list of figures]{ Scenario 1: General comparison (a) True and estimated disturbances (b) Estimation signals (c) Phase portrait for the SMC-SLDO }
\label{fig_2}
\end{figure}

\subsection{Scenario 2: Eliminating Chattering Effect}

\begin{figure}[h!]
\centering
\subfigure[ ]{
\includegraphics[width=0.99\textwidth]{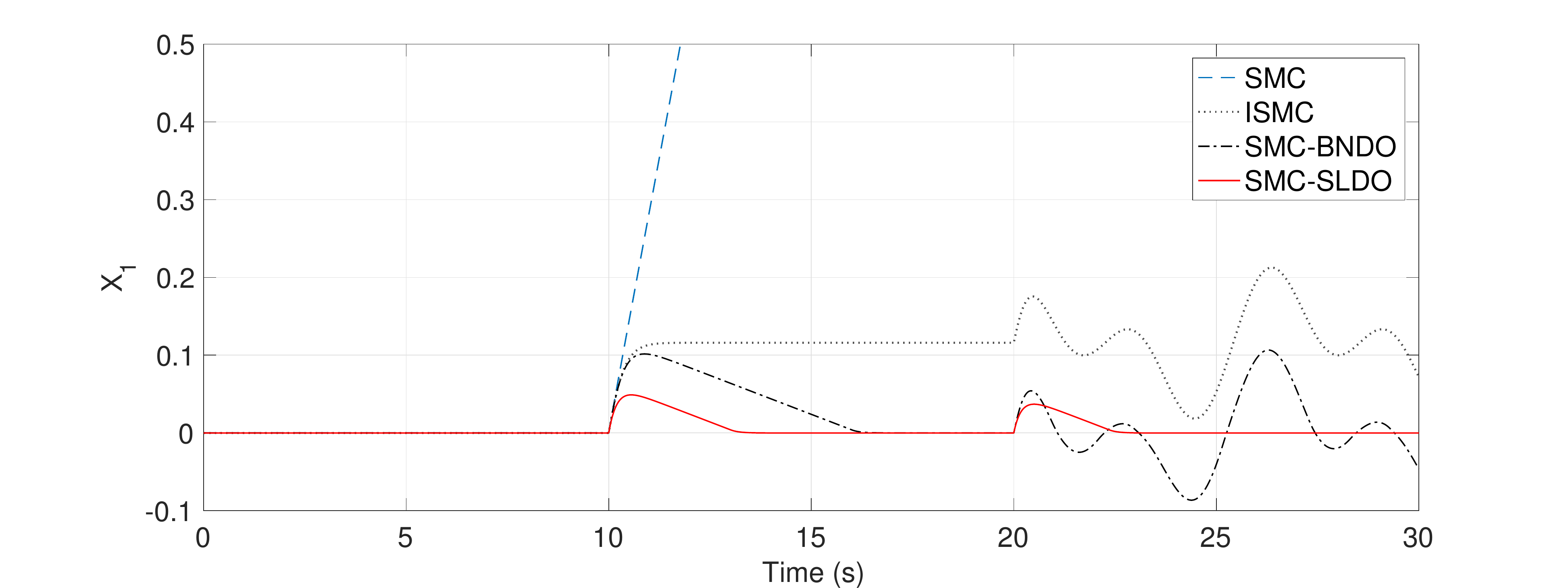}
\label{fig_2nd_x1}
}
\subfigure[ ]{
\includegraphics[width=0.99\textwidth]{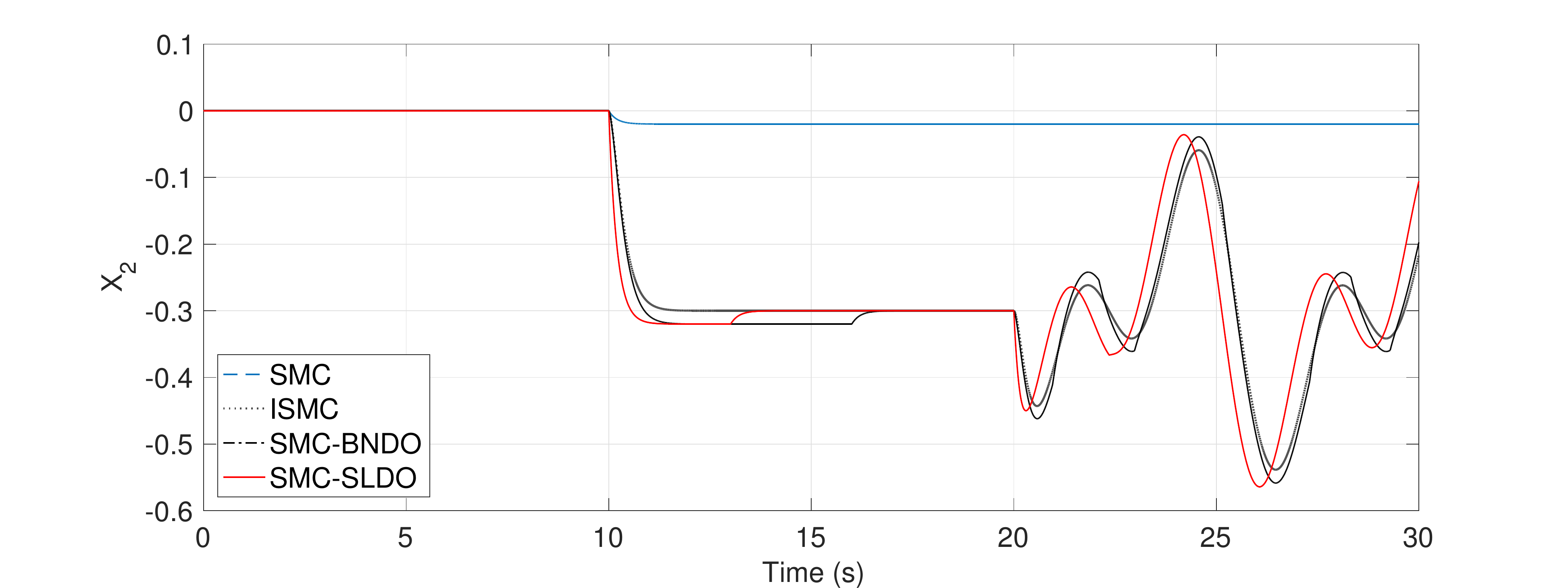}
\label{fig_2nd_x2}
}
\subfigure[ ]{
\includegraphics[width=0.99\textwidth]{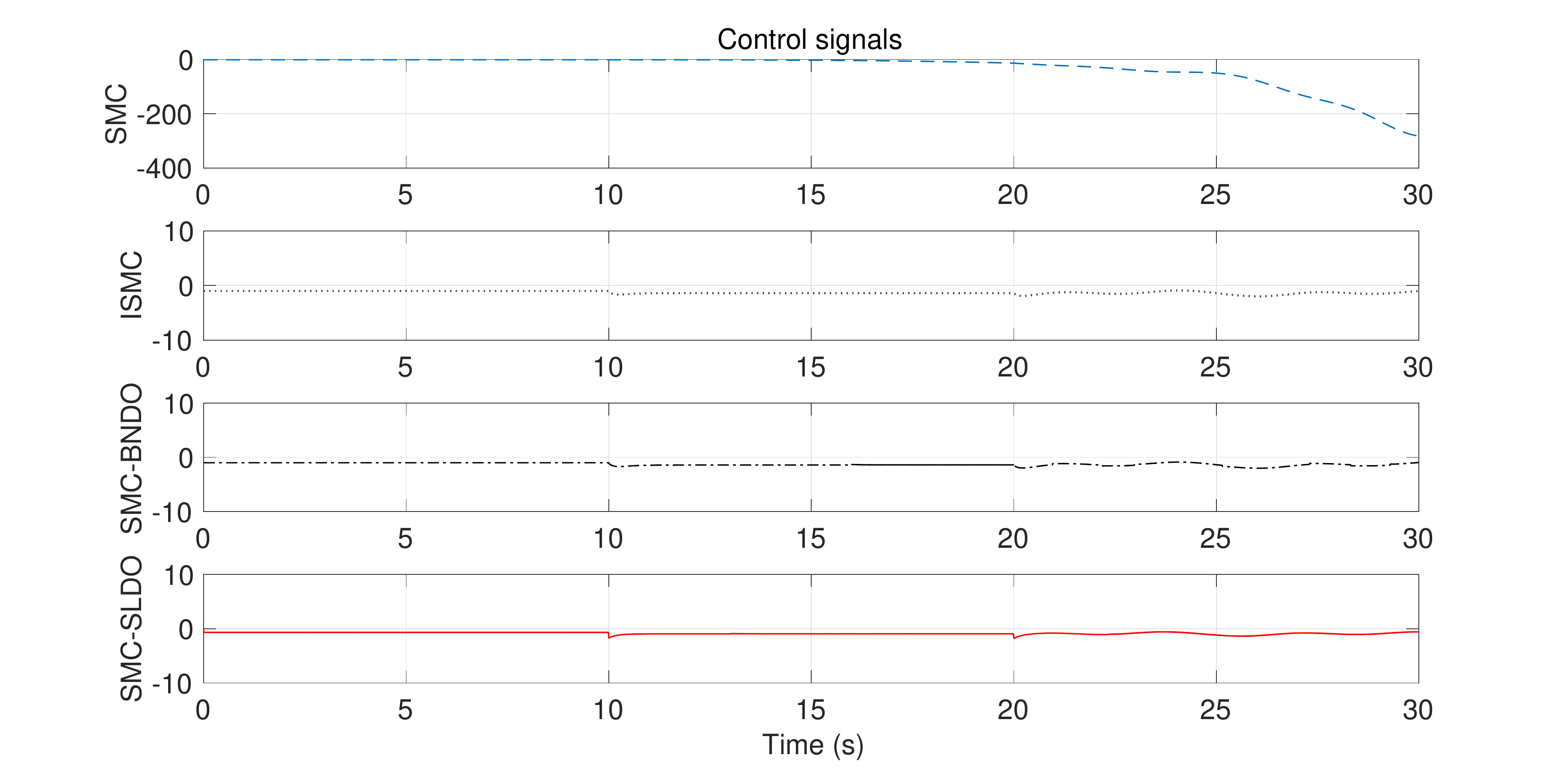}
\label{fig_2nd_control_signals}
}
\caption[Optional caption for list of figures]{ Scenario 2: Eliminating chattering effect: (a) Responses of state $x_{1}$ (b) Responses of state $x_{2}$ (c) Control signals }
\label{fig_3}
\end{figure}

To guarantee the stability of the system, the controller gains of the SMC, ISMC, SMC-BNDO and SMC-SLDO must be respectively set in a way that $k>\lambda d$, $k>2 \lambda d $,  $k> (\lambda+lz) e_{d}$ and  $k>\lambda e_{d}$ as stated respectively in Sections \ref{sec_smc},\ref{sec_ismc}, \ref{sec_smcbndo} and \ref{sec_smcsldo}. Since the controller gains of the SMC-BNDO and SMC-SLDO are required to be selected larger than the upper bound of the disturbance estimation error rather than the upper bound of the disturbance, and the disturbance estimation errors tend to converge to zero, the controller gains for the SMC-BNDO and SMC-SLDO can be selected smaller than the ones for the SMC and ISMC. This results in reducing the chattering effects \cite{Lu2009,xYang2013}. Therefore, the controller gains are set to $k=0.1$. 

The initial conditions on the states are set to $x(0)=[0 , 0]^{T}$. All other design parameters are same as the previous simulation scenario in Section \ref{sec_scenario1}.  As can be seen from Figs. \ref{fig_2nd_x1}-\ref{fig_2nd_x2}, since the controller gains of the SMC and ISMC are not larger than $\lambda d$, these controllers fail to remove mismatched uncertainty effects on the system. The SMC-BNDO and SMC-SLDO can drive the states to the desired state trajectory while chattering effects on their control signals are not observed as shown in Fig. \ref{fig_2nd_control_signals}. Moreover, the SMC-SLDO gives less settling time and overshoot than the SMC-BNDO in transient response in the presence of mismatched time-invariant disturbance. The reason is that since SMC-BNDO requires a larger controller gain than the SMC-SLDO to guarantee the stability as mentioned in the previous paragraph. For this reason, the controller gain of the SMC-BNDO must be selected larger than the one for the SMC-SLDO to obtain the same control performance. This can be interpreted as another superiority of the SMC-SLDO.  

The comparison of the SMC, ISMC, SMC-BNDO and SMC-SLDO methods are given in Table \ref{Tab}. Moreover, the mean values of the errors for the SMC, ISMC, SMC-BNDO and SMC-SLDO are respectively $1.8839$, $0.0775$, $0.0130$ and $0.049$ while the settling times of the SMC-BNDO and SMC-SLDO in the presence of time-invariant mismatched uncertainties are respectively $16.5$ and $13.5$ seconds. The results show that the SMC-SLDO satisfies all listed requirements while other methods satisfies just some of them. Furthermore, the SMC-SLDO ensures better performance even if others satisfy the requirements. 

\begin{table}[h!]
\centering
\caption{Comparison of SMC, ISMC, SMC-BNDO and SMC-SLDO methods.}
\label{Tab}
\begin{tabular}{l  c  c  c  c }
\hline
 & SMC & ISMC & SMC-BNDO & SMC-SLDO \\
 \hline
 \hline
Nominal performance & + & - & + & + \\
\hline
Robust control performance against & - & + & + & + \\
mismatched time-invariant disturbance &  &  &  &  \\
\hline
Robust control performance against & - & - & - & + \\
mismatched time-varying disturbance &  &  &  & \\
\hline
Eliminating chattering effect with robust control  & - & - & + & + \\
performance for mismatched  time-invariant  disturbance &  &  &  &  \\
\hline
Eliminating chattering effect with robust control & - & - & - & + \\
performance for mismatched time-varying disturbance &  &  &  & \\
\hline
\end{tabular}
\end{table}

\section{Conclusion}\label{sec_conc}

The SMC-SLDO has been elaborated in this paper to remove mismatched uncertainties effects on systems. The stability of the SMC-SLDO is proven and the simulation results demonstrate that the NFS running in parallel with the conventional estimation law can estimate time-varying disturbances. The SMC-SLDO ensures the robust control performance in the presence of mismatched time-varying disturbances when compared to the SMC, ISMC, and SMC-BNDO. Moreover, it provides less settling and rise times, and overshoot than SMC-BNDO in the presence of mismatched time-invariant disturbances while it remains the nominal control performance in the absence of mismatched disturbances. Through the online SMC theory-based learning algorithm, the parameters of the NFS are autonomously adjusted to learn disturbances, and the novel SMC algorithm can drive states to the desired state trajectory. Thus, the proposed SMC-SLDO makes the system robust against mismatched disturbances and also eliminates the chattering effects on the control signal significantly. 

\appendix
\section{Calculation of $\dot{\tau}_{n}$}\label{Sec_appendix}

The time-derivatives of \eqref{eq_mu1} and \eqref{eq_mu2} are obtained as:
\begin{eqnarray}
\dot{\mu}_{1i}(\xi_1) & = & -2 N_{1i} \dot{N}_{1i}\mu_{1i}(\xi_1) \nonumber \\
\dot{\mu}_{2j}(\xi_2) & = & -2 N_{2j} \dot{N}_{2j}\mu_{2j}(\xi_2)
\end{eqnarray}
where
\begin{eqnarray*}\label{eq_N1iN2j}
N_{1i}=\Big(\frac{\xi_1-c_{1i}}{\sigma_{1i}}\Big), \quad  \dot{N}_{1i}  =   \frac{(\dot{\xi_1} - \dot{c}_{{1i}})\sigma_{1i}-(\xi_1 - c_{1i})\dot{\sigma}_{1i}}{\sigma^2 _{1i}} \nonumber  \\
N_{2j}=\Big(\frac{\xi_2-c_{2j}}{\sigma_{2j}}\Big), \quad 
\dot{N}_{2i} = \frac{(\dot{\xi_2} - \dot{c}_{{2i}})\sigma_{2i}-(\xi_2 - c_{2i})\dot{\sigma}_{2i}}{\sigma^2 _{2i}} 
\end{eqnarray*}
It is obtained by combining $N_{1i}$,$ N_{2j}$,  $\dot{N}_{1i}$ and  $\dot{N}_{2j}$
\begin{equation}\label{eq_N1iN2j_N1iN2j_dot}
N_{1i}\dot{N}_{1i}=N_{2i}\dot{N}_{2i}=\alpha_{1} sgn{( \eta )}
\end{equation}

The time-derivative of \eqref{eq_wij_normalized} is obtained as follows:
\begin{eqnarray}
\dot{\widetilde{w}}_{ij} & = & \frac{\Big(\mu_{1i}(\xi_1) \mu_{2j}(\xi_2)\Big)'\Big(\sum_{i=1}^{I}\sum_{j=1}^{J}w_{ij}\Big)}{(\sum_{i=1}^{I}\sum_{j=1}^{J}w_{ij})^2}  - \frac{\big(w_{ij}\big)\Big(\sum_{i=1}^{I}\sum_{j=1}^{J}\mu_{1i}(\xi_1)   \mu_{2j}(\xi_2)\Big)'}{(\sum_{i=1}^{I}\sum_{j=1}^{J}w_{ij})^2}
\end{eqnarray}

Since $\widetilde{w}_{ij} = \frac{w_{ij}}{\sum_{i=1}^{I}\sum_{j=1}^{J}\underline{w}_{ij}}$,
\begin{eqnarray}\label{dotwij_normalized}
\dot{\widetilde{w}}_{ij} & = & \frac{ \dot{\mu}_{1i}(\xi_1) \mu_{2j}(\xi_2)+ \mu_{1i}(\xi_1) \dot{\mu}_{2j}(\xi_2)}{\sum_{i=1}^{I}\sum_{j=1}^{J}w_{ij}}  -\frac{\widetilde{w}_{ij}\sum_{i=1}^{I}\sum_{j=1}^{J} \Big(\dot{\mu}_{1i}(\xi_1) \mu_{2j}(\xi_2)+ \mu_{1i}(\xi_1) \dot{\mu}_{2j}(\xi_2)\Big)}{(\sum_{i=1}^{I}\sum_{j=1}^{J}w_{ij})} \nonumber \\
& = &  \frac{-2 \Big(N_{1i} \dot{N}_{1i} + N_{2j} \dot{N}_{2j} \Big) \mu_{1i}(\xi_1)\mu_{2j}(\xi_2)}{\sum_{i=1}^{I}\sum_{j=1}^{J}w_{ij}}  -\frac{\widetilde{w}_{ij}\sum_{i=1}^{I}\sum_{j=1}^{J} \Big(-2 \big(  N_{1i} \dot{N}_{1i} +  N_{2j} \dot{N}_{2j} \big) \mu_{1i}(\xi_1)\mu_{2j}(\xi_2)\Big)}{ \sum_{i=1}^{I}\sum_{j=1}^{J}w_{ij}}  \nonumber \\
&=& -\widetilde{w}_{ij}\dot{K}_{ij} + \widetilde{w}_{ij} \sum_{i=1}^{I}\sum_{j=1}^{J}\widetilde{w}_{ij} \dot{K}_{ij} 
\end{eqnarray}
where
\begin{equation*}
\dot{K}_{ij} = 2 \Big(N_{1i}\dot{N}_{1i}+N_{2j} \dot{N}_{2j}\Big)
\end{equation*}

The time-derivative of \eqref{eq_taun} is obtained to find $\dot{\tau_n}$ as follows:
\begin{equation}
\dot{\tau_n} = \sum_{i=1}^{I}\sum_{j=1}^{J}(\dot{f}_{ij} \widetilde{w}_{ij}+ f_{ij}\dot{\widetilde{w}}_{ij})
\end{equation}

If \eqref{dotwij_normalized} is inserted to the aforementioned equation:
\begin{eqnarray}\label{eq_taun_dot}
\dot{\tau_n} =  \sum_{i=1}^{I}\sum_{j=1}^{J}\bigg(\Big(-\widetilde{w}_{ij}\dot{K}_{ij}+\widetilde{w}_{ij} \sum_{i=1}^{I}\sum_{j=1}^{J}\widetilde{w}_{ij} \dot{K}_{ij} \Big)f_{ij}+\dot{f}_{ij}\ \widetilde{w}_{ij} \bigg)
\end{eqnarray}

\eqref{eq_N1iN2j_N1iN2j_dot} is inserted to \eqref{eq_taun_dot}:
\begin{eqnarray}\label{eq_taun_dot_2}
\dot{\tau_n} &=&  \sum_{i=1}^{I}\sum_{j=1}^{J}  -4\alpha_{1} sgn{(s)} \Big(\widetilde{w}_{ij} f_{ij} - \widetilde{w}_{ij} f_{ij} \sum_{i=1}^{I}\sum_{j=1}^{J}\widetilde{w}_{ij} \Big) \nonumber \\
 &&  +\widetilde{w}_{ij}\dot{f}_{ij} 
\end{eqnarray}

Since $\sum_{i=1}^{I}\sum_{j=1}^{J}\widetilde{w}_{ij} =1$, \eqref{eq_taun_dot_2} becomes by using \eqref{eq_f_ij} as follows:
\begin{eqnarray}\label{eq_taun_dot_3}
\dot{\tau_n}  =  \sum_{i=1}^{I}\sum_{j=1}^{J}\Big(\widetilde{w}_{ij} \dot{f}_{ij} \Big) =  - \alpha_{2} sgn{( \eta )}
\end{eqnarray}

\bibliographystyle{SageH} 
\bibliography{reference.bib}

\end{document}